\documentclass[journal,twoside,web]{ieeecolor}

\usepackage{generic}
\usepackage{cite}
\usepackage{amsmath,amssymb,amsfonts}
\usepackage{algorithmic}
\usepackage{graphicx}
\usepackage{textcomp}
\usepackage{epsfig} 
\usepackage{lcsys}

\usepackage{url}
\usepackage{makecell}

\newtheorem{theorem}{Theorem}

\newtheorem{definition}{Definition}
\newtheorem{corollary}{Corollary}

\begin{document}

\title{Trapping Regions for Quadratic Systems with Generalized Lossless Nonlinearities}

\author{Diganta Bhattacharjee$^{a}$, Shih-Chi Liao$^{b}$, Peter J. Seiler$^{b}$, and Maziar S. Hemati$^{a}$
\thanks{$^{a}$Department of Aerospace Engineering and Mechanics, University of Minnesota, Minneapolis, MN 55455, USA; \texttt{dbhattac@umn.edu, mhemati@umn.edu}}
\thanks{$^{b}$Department of Electrical Engineering and Computer Science, University of Michigan, Ann Arbor, MI 48109, USA; \texttt{shihchil@umich.edu, pseiler@umich.edu}}
} 

\maketitle
\thispagestyle{empty}

\begin{abstract}
We consider a class of quadratic systems, primarily motivated by incompressible fluid flows, where the nonlinearities are generalized lossless: they do not produce or dissipate energy, as measured by a generalized quadratic metric. 
Our goal is to compute trapping regions, which are forward invariant sets that certify ultimate boundedness.
The key contribution is a novel parameterization of the generalized lossless condition that enables optimization of trapping regions for a broader class of quadratic systems.
We also formulate the conditions for ellipsoidal trapping regions, whereas spherical regions have been the focus of prior works. 
We provide three numerical examples, which 
demonstrate the improvements offered by the proposed approach relative to existing methods.
\end{abstract}

\begin{IEEEkeywords}
Lyapunov methods; LMIs; Optimization
\end{IEEEkeywords}

\section{Introduction}
\IEEEPARstart{W}{e} consider the class of autonomous, nonlinear ordinary differential equations with a constant, linear, and quadratic term.  
This class of systems can be used to model physical processes, e.g., incompressible fluid flows \cite{schlegel2015long, callaham2023multiscale, loiseau2018constrained, cavalieri2022transition, balajewicz2013low, kraichnan1989there}, atmospheric processes \cite{goyal2025guaranteed, peng2025extending}, and magnetohydrodynamics \cite{kaptanoglu2021physics, kaptanoglu2021promoting}. 
An important feature for these systems is (long-term) ultimate boundedness of state-trajectories \cite{vidyasagar2002nonlinear, khalil2002nonlinear}.  
We focus on the specific notion of a globally attracting trapping region as this is commonly used to assess boundedness of incompressible fluid flow models \cite{schlegel2015long, kaptanoglu2021promoting, goyal2025guaranteed, liao2025convex, heide2025data}.
It is important to note that the quadratic term in such models is lossless: the inner product between system state and quadratic nonlinearity is identically equal to zero \cite{schmid_henningson_book}.  
This implies that the quadratic term neither produces nor dissipates energy. 
Prior work has used this lossless condition and Lyapunov theory to derive sufficient conditions for the existence of a trapping region \cite{schlegel2015long, kaptanoglu2021promoting, goyal2025guaranteed, liao2025convex, heide2025data}.

Our key contribution is the incorporation of a generalized lossless condition for the analysis. 
Specifically, the generalized condition is: the \emph{weighted} inner product between system state and quadratic nonlinearity is identically equal to zero. 
We provide a novel parameterization of this generalized lossless condition.
We then formulate conditions for an ellipsoidal trapping region subject to a linear equality constraint based on our parameterization. 
This enables analysis of a broader class of quadratic systems, e.g., incompressible fluid flow models with a wide range of boundary conditions and geometries \cite{schlegel2015long, callaham2023multiscale, loiseau2018constrained, cavalieri2022transition, balajewicz2013low, kraichnan1989there}.   
A numerical procedure is proposed that simultaneously optimizes the trapping ellipsoid and the generalized lossless constraint. 
We demonstrate our proposed framework on three numerical examples, two of which admit the standard lossless condition used in prior work. 
For those two examples, our approach reduces conservatism in the ultimate bound relative to existing methods \cite{schlegel2015long, goyal2025guaranteed, liao2025convex}.
The third example does not satisfy the standard lossless constraint, making existing state-of-the-art methods inapplicable.

There is a large relevant literature on trapping regions for quadratic systems \cite{schlegel2015long, goyal2025guaranteed, liao2025convex, kaptanoglu2021promoting, heide2025data, peng2025extending}.
This is due to the importance for analyzing ultimate boundedness of fluid flows. 
Ultimate boundedness can be certified via invariant sets such as trapping regions. Other methods include region of attraction computation using sum-of-squares programming \cite{tan2008stability} and linear matrix inequalities \cite{chesi2005lmi} as well as Lur'e-type interpretation of the quadratic dynamics \cite{liu2020input}.
The most relevant work is the recent data-driven modeling framework by Goyal et al. \cite{goyal2025guaranteed}. 
They introduced the concept of a generalized lossless condition that we use in our results. 
However, they reverted to the standard lossless constraint for all their numerical results. 
In fact, they remark that future work will ``require developing novel strategies to efficiently handle the additional complexity introduced by'' the generalized lossless condition.  
Our work addresses this gap.

\textit{Notation:} We use $P > 0$ (or $P < 0$) to mean that $P$ is a symmetric, positive (or negative) definite matrix.
We define an ellipsoid centered at the origin with shape $P > 0$ and radius $r>0$ as $\mathcal{E}(P,r) := \{ x \in \mathbb{R}^n \, : \, x^\top P x \leq r^2\}$.
A norm ball of radius $r>0$ is denoted as $\mathcal{B}(r):=\{x \in \mathbb{R}^n \, : \, x^\top x \leq r^2\}$.
If $X \in \mathbb{R}^{n\times n}$ then $\mathrm{vec}(X)$ denotes the $n^2 \times 1$ vector obtained by vertically stacking the columns of $X$.
Moreover, $\lambda_{\mathrm{min}}(X)$ and $\lambda_{\mathrm{max}}(X)$ denote the minimum and maximum eigenvalues of $X$, respectively.


\section{Preliminaries} \label{sec:Preliminaries}
This section provides the technical background and problem formulation for the trapping region analysis. 
\subsection{Quadratic Systems} \label{Sec:Class of Models}
Consider a quadratic system of the form 
\begin{equation} \label{eq:Quad_sys}
\dot{x}(t) = Ax(t) + Q(x(t)) + d
\end{equation} 
where $x(t) \in \mathbb{R}^n$ is the state, $d \in \mathbb{R}^n$ is a constant bias vector, and $A \in \mathbb{R}^{n \times n}$ is the state matrix. 
Moreover, $Q: \mathbb{R}^n \rightarrow \mathbb{R}^n$ is a quadratic nonlinearity defined as 
\begin{equation} \label{eq:Quadratic_nonlinearity}
Q(x) := \begin{bmatrix}
x^\top Q_1 x & x^\top Q_2 x & \cdots & x^\top Q_n x
\end{bmatrix}^\top
\end{equation}
where $Q_i = Q_i^\top \in \mathbb{R}^{n \times n}, i=1,2,\dots,n$, i.e., each $Q_i$ is symmetric but not necessarily sign definite.
We can analyze systems of the form \eqref{eq:Quad_sys} with an uncertain $d$, but we take $d$ as a constant to be consistent with the trapping region literature. 
We take a coordinate shift $m \in \mathbb{R}^n$ \cite{schlegel2015long, liao2025convex} into account. 
Specifically, we express the system \eqref{eq:Quad_sys} in a coordinate system $y := x - m$ as \cite{liao2025convex}
\begin{equation} \label{eq:Quadsys_shift_coordinates}
\dot{y}(t) = L(m) y(t) + Q(y(t)) + c(m)
\end{equation} 
where $Q(\cdot)$ is as given in \eqref{eq:Quadratic_nonlinearity} and 
\begin{align*}
L(m) &:= A + 2 \begin{bmatrix}
Q_1 m^\mathrm & Q_2 m^\mathrm & \cdots & Q_n m
\end{bmatrix}^\top, \\
c(m) &:= d + Am + Q(m).
\end{align*}
Prior results in the literature utilize an energy-preserving or lossless property on the nonlinearity, which is given by $y^\top Q(y) = 0, \forall y \in \mathbb{R}^n$ \cite{schlegel2015long, liao2025convex, kaptanoglu2021promoting}.
However, this particular form of the constraint might not hold in all applications.
This motivates a generalized form of the lossless property in this letter and it is defined next. 
\begin{definition} \label{definition:lossless_nonlinearity}
A quadratic nonlinearity $Q$ is termed \emph{generalized lossless} if there exists some $S \in \mathbb{R}^{n \times n}$ such that $y^\top S Q(y) = 0, \forall y \in \mathbb{R}^n$.
\end{definition} 
Setting $S=I$ retrieves the lossless constraint used in the prior work on trapping region analysis.
The generalized lossless constraint in Definition \ref{definition:lossless_nonlinearity} has been proposed in the context of low-order flow modeling techniques \cite{goyal2025guaranteed}.
The generalization not only extends applicability of the analysis to a broader class of quadratic systems, but also facilitates an optimization over the admissible set of lossless inner product matrices $S$ for a given quadratic nonlinearity $Q$.

As an example, consider the chaotic Lorenz attractor which has been used as a benchmark problem in the prior work \cite{schlegel2015long, kaptanoglu2021promoting}. 
The dynamics are governed by
\begin{equation} \label{eq:Lorenz_system}
\begin{bmatrix}
\dot{x}_1 \\ \dot{x}_2 \\ \dot{x}_3
\end{bmatrix} = \begin{bmatrix}
-\sigma & \sigma & 0 \\ 
\rho & -1 & 0 \\
0 & 0 & -\eta
\end{bmatrix} \begin{bmatrix}
x_1 \\ x_2 \\ x_3
\end{bmatrix} + \begin{bmatrix}
0 \\ -x_1 x_3 \\ x_1 x_2
\end{bmatrix}
\end{equation}
where $\sigma, \rho, \eta$ are positive constants.
The nonlinearity here is generalized lossless with respect to any matrix
\begin{equation} \label{eq:generalized_lossless_matrix_Lorenz}
    S = \begin{bmatrix}
        s_{11}  & 0 & 0 \\
        s_{21} & s_{22} & 0 \\
        s_{31} & 0  & s_{22}
    \end{bmatrix},~s_{ij} \in \mathbb{R}.
\end{equation} 
Prior work on trapping regions for the Lorenz attractor utilized $S=I$ \cite{schlegel2015long, liao2025convex, kaptanoglu2021promoting}, which is a special case of the generalized form in \eqref{eq:generalized_lossless_matrix_Lorenz}. 
Our main result incorporates the generalized lossless constraint into a sufficient trapping region condition.

\subsection{Trapping Regions}
 Let $\phi(t;y_0,t_0)$ denote the solution of \eqref{eq:Quadsys_shift_coordinates} at time $t \geq t_0$ starting from the initial condition $y_0 \in \mathbb{R}^n$ at $t_0 \in \mathbb{R}_{\geq 0}$.
The main objective of this work is to establish global boundedness of the solutions of \eqref{eq:Quadsys_shift_coordinates}. 
This is defined as follows.
\begin{definition} [\cite{khalil2002nonlinear, liao2025convex}]
The solutions of system \eqref{eq:Quadsys_shift_coordinates} are \emph{globally uniformly ultimately bounded} if there exists a scalar $\beta > 0$, independent of $t_0 \geq 0$, and a function $T: \mathbb{R}_{\geq 0} \rightarrow \mathbb{R}_{\geq 0}$ such that $\| \phi (t;y_0,t_0)\| \leq \beta$ for all $y_0 \in \mathbb{R}^n$ and $t \geq t_0 + T(\|y_0\|)$.
\end{definition}
Note that a system is bounded  in the original coordinates $x$ if and only if it is bounded in the shifted coordinates $y=x-m$. Our goal is to use the concept of a trapping region to show boundedness of the quadratic dynamics. A trapping region is defined as follows.
\begin{definition}
A \emph{trapping region} $\mathcal{D} \subseteq \mathbb{R}^n$ is a compact set that is forward invariant, i.e., if $y(t_0) \in \mathcal{D}$ then $y(t) \in \mathcal{D},~\forall t \geq t_0$. A trapping region $\mathcal{D} \subseteq \mathbb{R}^n$ is called \emph{globally attracting} if there exists an open set $\mathcal{O} \supset \mathcal{D}$ such that every trajectory enters $\mathcal{O}$ for sufficiently large $t$.
\end{definition} 
By definition, if a trapping region is globally attracting then every trajectory will eventually enter this set, possibly in the limit as $t\to\infty$. 
Moreover, if a trajectory enters the trapping region, then it remains in the trapping region for all future time.  
As one special case, consider $\mathcal{D} := \{0\}$. 
In this case, $\mathcal{D}$ is a trapping region if and only if $y=0$ is an equilibrium point. 
In addition, $\mathcal{D} := \{0\}$ is a globally attracting trapping region if and only if $y=0$ is a globally asymptotically stable equilibrium point.

Previous work on quadratic systems \cite{schlegel2015long, liao2025convex} has focused on spherical trapping regions $B(r)$ with ``energy'' functions $W(y):=y^\top y$.  
Suppose $W$ is monotonically decreasing outside a given ball $B(r)$. 
Then it can be shown, via a Lyapunov-like analysis, that the ball is a trapping region. 
In this special case, the ball $B(r)$ is known as a globally monotonically attracting trapping region \cite{schlegel2015long,  liao2025convex}.   
In this letter, we will consider more general (ellipsoidal) trapping regions and more general classes of Lyapunov-like (``energy'') functions of the form $V(y) := y^\top P y$ where $P > 0$.   
We focus on globally attracting trapping regions in this work. 
Throughout the remainder of the letter, we will be using the terms \emph{trapping region} and \emph{trapping ellipsoid} in reference to a globally attracting trapping region.


\section{Trapping Region Analysis} \label{sec: Global Trapping Region Analysis}
Technical results of the proposed analysis as well as the associated numerical framework for computing the trapping ellipsoids are discussed in this section.
\subsection{Analysis Conditions} \label{Sec: Analysis Conditions}
We first establish the condition for which a trapping region exists under the generalized lossless constraint.
\begin{theorem} \label{theorem:sufficient_trapping_condition}
    Consider system \eqref{eq:Quadsys_shift_coordinates} where the quadratic term $Q$ satisfies the generalized lossless constraint for a matrix $S > 0$.
    The system admits a trapping region if there exists $m \in \mathbb{R}^n$ such that $S L(m) + L^\top(m) S < 0$.
\end{theorem}
\begin{proof}
Consider energy function $V(y) = y^\top P y$ with $P > 0$ and define $q(t) := Q(y(t))$ to simplify notations. 
Taking the derivative of the energy function along the trajectories of system \eqref{eq:Quadsys_shift_coordinates} yields
\begin{equation*}
\begin{split}
    \dot{V}(y(t)) &= y^\top(t) \left(PL(m) + L^\top(m)P   \right) y(t) \\
    & \ + y^\top(t) P q(t) + q^\top(t) P y(t) \\ & + y^\top(t) P c(m) + c^\top(m) P y(t).
    \end{split}
\end{equation*}
Existence of a (global) trapping region necessitates $\dot{V}(y) <0$ globally outside the trapping region. 
Thus, the cubic terms $y^\top(t) P q(t)$ and $q^\top(t) P y(t)$ in the above expression should exactly equate to zero. 
To this end, recall that $Q$ is generalized lossless with respect to $S > 0$ by assumption.
Hence, if $P = S$ then we have $y^\top(t) P q(t) = q^\top(t) P y(t) = 0$ (cf. Definition \ref{definition:lossless_nonlinearity}) and the energy derivative becomes
\begin{equation} \label{eq:energy_derivative_globalTRA-1}
\begin{split}
    \dot{V}(y(t)) &= y^\top(t) \left(S L(m) + L^\top(m) S  \right) y(t) \\
    & \ + y^\top(t) S c(m) + c^\top(m) S y(t).
    \end{split}
\end{equation}
If the quadratic term is negative for all $y \in \mathbb{R}^n \setminus\{0\}$, then $\dot{V}(y) < 0$ holds for sufficiently large $y$. This is true if $S L(m) + L^\top(m) S < 0$ for some $m \in \mathbb{R}^n$.
Thus, there exist invariant sets (i.e., trapping regions) in the form of sub-level sets of $V(y)$, which completes the proof.
\end{proof}
The sufficient condition in Theorem \ref{theorem:sufficient_trapping_condition} is similar to Corollary 5.1 in Goyal et al.\cite{goyal2025guaranteed}. 
The latter also characterizes a norm ball as a trapping region. 
Define $\lambda_1 := \lambda_{\mathrm{max}}\left( S L(m) + L^\top(m) S \right)$ and let the conditions stated in Theorem \ref{theorem:sufficient_trapping_condition} hold. 
Then, Corollary 5.1 in Goyal et al. \cite{goyal2025guaranteed} yields $\mathcal{B}(r)$ with $r=\frac{2 \|S c(m)\|}{|\lambda_1|}$ as a trapping region for system \eqref{eq:Quadsys_shift_coordinates}. Note that the expression here is slightly different from that in \cite{goyal2025guaranteed} as there is a minor typo in \cite{goyal2025guaranteed}.
However, this norm ball is not necessarily invariant with respect to the dynamics, hence not technically a trapping region.
We resolve this inconsistency by considering ellipsoidal sub-level sets of the energy function as trapping regions in this letter.
Substituting $S=I$ in Theorem \ref{theorem:sufficient_trapping_condition} recovers the standard condition for the existence of a globally monotonically attracting spherical trapping region in prior literature \cite{schlegel2015long,liao2025convex}.   

A quadratic nonlinearity $Q$ satisfies the generalized lossless constraint in Definition \ref{definition:lossless_nonlinearity} for many possible $S$ (see, for example, the discussion on the Lorenz attractor in Section \ref{Sec:Class of Models}). 
We utilize a polynomial approach for characterizing all of these possible generalized lossless matrices for a given quadratic nonlinearity.
This involves recasting the generalized lossless constraint into a form that depends only on the matrices defining the quadratic nonlinearity (see \eqref{eq:Quadratic_nonlinearity}) and the generalized lossless matrix $S$. 
First, we note that the inner product $y^\top S Q(y)$ is a cubic polynomial in $y$ but is linear in $S$.
Hence this expression can be written as $g(S) h(y)$ where $g(S)$ is a row vector that is linear in $S$ and $h(y)$ is a column vector of all cubic monomials in $y\in \mathbb{R}^n$. 
The generalized lossless condition $y^\top S Q(y) = g(S) h(y) =0$ for all $y\in \mathbb{R}^n$ thus requires $g(S)=0$.  
Moreover, there exists a matrix $G$ such that $g(S) = G^\top \mathrm{vec}(S)$ since $g$ is a linear function of $S$.
This yields the necessary constraint for the generalized lossless property: $G^\top \mathrm{vec}(S)=0$.
Next, we state a result that combines the sufficient condition in Theorem \ref{theorem:sufficient_trapping_condition} with the above characterization of the generalized lossless property.
\begin{corollary} \label{corollary:sufficient_trapping_condition_genloss}
Consider system \eqref{eq:Quadsys_shift_coordinates} where the quadratic term $Q$ satisfies the generalized lossless constraint for all matrices $S$ that satisfy $G^\top \mathrm{vec}(S)=0$. The system admits a trapping region if there exists $m \in \mathbb{R}^n$ and $P > 0$ such that $P L(m) + L^\top(m) P < 0$ and $G^\top \mathrm{vec}(P) = 0$.
\end{corollary}
\begin{proof}
Follows directly from the proof of Theorem \ref{theorem:sufficient_trapping_condition} with $P = S > 0$ and the above reformulation of the generalized lossless constraint.
\end{proof}
The next result characterizes a trapping ellipsoid based on the sufficient conditions in Corollary \ref{corollary:sufficient_trapping_condition_genloss}. Define
\begin{equation*}
   \Theta (m,P,r,\chi) := \begin{bmatrix}
    PL(m) + L^\top(m)P + \chi P & P c(m) \\
c^\top(m) P & -\chi r^2
\end{bmatrix}
\end{equation*}
to simplify the notation.
\begin{theorem} \label{theorem:trapping_quadsys_lossless_shift_coordinates}
Consider system \eqref{eq:Quadsys_shift_coordinates} where the quadratic term $Q$ satisfies the generalized lossless constraint for all matrices $S$ that satisfy $G^\top \mathrm{vec}(S)=0$.  
Then, the ellipsoid $\mathcal{E}(P, r)$ is a trapping region if there exists $m \in \mathbb{R}^n$, $P > 0$, $r$, $\chi \geq 0$ such that 
\begin{align}
\Theta (m,P,r,\chi) < 0, \  G^\top \mathrm{vec}(P) = 0. \label{eq:global_trapping_constraints}
\end{align} 
\end{theorem}
\begin{proof}
Similar to the proof of Theorem \ref{theorem:sufficient_trapping_condition}, consider
$V(y) = y^\top P y$ where $P > 0$ and take $q(t) = Q(y(t))$. 
We want to characterize the ellipsoid $\mathcal{E}(P, r)$, which equals the sub-level set $\{ y\in \mathbb{R}^n: \, V(y) \leq r^2 \}$, as the trapping region.
Considering the quadratic nonlinearity $Q$ to be generalized lossless with respect to $S = P$, we have $G^\top \mathrm{vec}(P) = 0$.
Then, $\dot{V}(y(t))$ along trajectories of \eqref{eq:Quadsys_shift_coordinates} is as shown in \eqref{eq:energy_derivative_globalTRA-1} after replacing $S$ with $P$. 
Now, we require $\dot{V}(y(t)) < 0$ for all $y(t) \in \mathbb{R}^n \setminus \mathcal{E} (P, r)$.
This holds if there exists $\chi \geq 0$ such that
\begin{equation*}
    \dot{V}(y) + \chi \left( y^\top(t) P y(t) - r^2 \right) < 0, \forall y(t) \in \mathbb{R}^n \setminus \mathcal{E} (P, r).
\end{equation*} 
Substituting the expression of $\dot{V}(y)$ yields 
\begin{equation*} 
\begin{split}
    & y^\top(t) \left(PL(m) + L^\top(m)P + \chi P \right) y(t) \\ & + y^\top(t) P c(m) + c^\top(m) P y(t) - \chi r^2 < 0 
\end{split}
\end{equation*}
for all $y(t) \in \mathbb{R}^n \setminus \mathcal{E} (P, r)$. This is equivalent to
\begin{equation*}
\begin{bmatrix}
        y(t) \\ 1
    \end{bmatrix}^\top 
    \Theta (m,P,r,\chi)  \begin{bmatrix}
        y(t) \\ 1
    \end{bmatrix} < 0,\forall y(t) \in \mathbb{R}^n \setminus \mathcal{E} (P, r)
\end{equation*}
which holds if $\Theta (m,P,r,\chi) < 0$.
Hence, we have $\dot{V}(y) < 0$ on the boundary and outside of $\mathcal{E}(P, r)$. 
It is straightforward to show that $\mathcal{E}(P, r)$ is invariant with respect to the dynamics in \eqref{eq:Quadsys_shift_coordinates} (e.g., by application of Nagumo's theorem \cite{ames2019control}) and thus a trapping region. 
\end{proof}
If the constraints \eqref{eq:global_trapping_constraints} hold, the trajectories of \eqref{eq:Quadsys_shift_coordinates} satisfy $y(t) \in \mathcal{E}(P,r)$ for sufficiently large $t$. This translates into $x(t) \in \{m\} + \mathcal{E}(P,r)$ for the trajectories of \eqref{eq:Quad_sys}.

\subsection{Numerical Solutions} \label{sec:Numerical Solutions}
Any feasible tuple $(m,P,r,\chi)$ satisfying the constraints \eqref{eq:global_trapping_constraints} 
characterizes a trapping ellipsoid. 
The goal of our numerical procedure is to compute the tightest trapping ellipsoid $\mathcal{E} (P,r)$. 
This can be achieved by solving an optimization of the form
\begin{equation}
 \label{eq:nonlinear_trapping_optimization_QuadROM}
\begin{aligned}
    & \min_{m \in \mathbb{R}^n, P > 0, r, \chi \geq 0} && f(m,P,r) \\
    & \mathrm{subject~to} && \eqref{eq:global_trapping_constraints}
\end{aligned}
\end{equation}
where $f(\cdot)$ is an appropriately chosen objective function. For example, $f(m,P,r) = r^2$ optimizes the trapping ellipsoid radius.
Alternatively, one can optimize the ellipsoid shape $P$ by choosing an appropriate objective function (e.g., trace or logarithmic determinant of $P$).
Note that \eqref{eq:nonlinear_trapping_optimization_QuadROM} is a nonlinear optimization problem as $\Theta (m,P,r,\chi) < 0$ is a nonlinear constraint.
Hence, it can be solved in various ways. 
In this letter, we split the nonlinear optimization up into two main sub-problems.

The sufficient conditions in Corollary \ref{corollary:sufficient_trapping_condition_genloss} are used for computing the shift coordinates. 
The associated optimization problem is given by 
\begin{equation} \label{eq:shift_optimization_initial}
\begin{aligned}
& \min_{a \in \mathbb{R}, m \in \mathbb{R}^n, P > 0} &&  a  \\
& \mathrm{subject~to} && P L(m) + L^\top(m) P \leq a I, \\ 
& && \|m\| \leq \delta_m, \ G^\top \mathrm{vec}(P) = 0
\end{aligned}
\end{equation}
where $\delta_m$ is a pre-specified upper bound on the shift coordinate norm. 
This constraint essentially acts as a regularizer so that the solution remains within a region of interest in the state space.
The optimization problem \eqref{eq:shift_optimization_initial} is non-convex as the constraint $P L(m) + L^\top(m) P \leq a I$ is bilinear in the variables $(m, P)$. 
In this work, we solve \eqref{eq:shift_optimization_initial} by iterating over $m$ and $P$: solve for $P$ with a fixed $m$, then fix $P$ to solve for $m$. We initialize with a random, unit-norm $m$ and repeat until convergence in~$a$.

Let us denote the solution of \eqref{eq:shift_optimization_initial} by $(a^\star, m^\star, P^\star)$. 
We assume that $a^\star < 0$ such that a trapping region exists (see Corollary \ref{corollary:sufficient_trapping_condition_genloss}). 
Utilizing the associated shift coordinates $m^\star$, we formulate the following semi-definite program~(SDP) to compute the trapping ellipsoids:
\begin{equation} \label{eq:trapping_optimization}
\begin{aligned}
    & \min_{P > 0, r^2, \chi \geq 0}  && r^2 \\
    & \mathrm{subject~to} && 
    \Theta (m^\star,P,r,\chi) \leq 
-\epsilon I,  \\ 
& && G^\top \mathrm{vec}(P) = 0
\end{aligned}
\end{equation}
where $\epsilon>0$ is a pre-specified tolerance and we grid over $\chi$ to address the bilinear term $\chi P$ in $\Theta (m^\star,P,r,\chi)$.  
Let $(m^{\star^{[1]}} = m^\star, P^{\star^{[1]}}, r^\star, \chi^\star)$ be the ``best'' solution on the $\chi$ grid with the tightest ellipsoid.
To refine $(r^\star, \chi^\star)$, solve a generalized eigenvalue problem~(GEVP)~\cite{boyd1993method},
\begin{equation} \label{eq:trapping_optimization_GEVP}
\begin{aligned}
    & \min_{r^2, \chi \geq 0}  && r^2 \\
    & \mathrm{subject~to} && 
    r^2 \begin{bmatrix}
    0 & 0 \\
    0 & \chi
\end{bmatrix} - \Theta (m^{\star^{[1]}},P^{\star^{[1]}},0,\chi) > 0, 
\end{aligned} 
\end{equation} 
to obtain $(r^{\star^{[1]}}, \chi^{\star^{[1]}})$, yielding the 
solution tuple $\Psi_1 :=(m^{\star^{[1]}}, P^{\star^{[1]}}, r^{\star^{[1]}}, \chi^{\star^{[1]}})$.
A local search in the shift coordinates can be conducted to further refine $\Psi_1$. 
It is a two-step procedure, summarized as follows: 
\begin{enumerate}
    \item Set $m = m^{\star^{[1]}} + \Delta m$ and linearize $c(m)$ at $m^{\star^{[1]}}$; solve an SDP of the form \eqref{eq:trapping_optimization} with $P = P^{\star^{[1]}}, \chi = \chi^{\star^{[1]}}$ and $(\Delta m, r^2)$ as the variables to get $\Delta m^\star$.
    \item Set $m^{\star^{[2]}} = m^{\star^{[1]}} + \Delta m^\star$. Solve SDP \eqref{eq:trapping_optimization} with $( m^\star,\chi) = ( m^{\star^{[2]}}, \chi^{\star^{[1]}})$ to compute $(P^{\star^{[2]}}, r^{\star^{[2]}})$.
    Define $\Psi_2 := (m^{\star^{[2]}}, P^{\star^{[2]}}, r^{\star^{[2]}}, \chi^{\star^{[1]}})$.
\end{enumerate}
If $\mathcal{E}(P^{\star^{[2]}}, r^{\star^{[2]}}) \subset \mathcal{E}(P^{\star^{[1]}}, r^{\star^{[1]}})$, then $\Psi_2$ is the final solution and vice versa.  
We adopt the following two ways of treating the equality constraint $G^\top \mathrm{vec}(P) = 0$:
\begin{itemize}
    \item \textit{Hard constraint approach}: Explicitly enforce the constraint by assigning structure to $P=P^\top$ such that $G^\top \mathrm{vec}(P) = 0$ (see Sections \ref{Sec: Modified Lorenz–Stenflo System}, \ref{Sec: Two-dimensional Academic Example}), restrict $P$ to this structure when solving the SDPs and discard the equality constraint. 
    \item \textit{Soft constraint approach}: Solve the SDPs with $G^\top \mathrm{vec}(P) = 0$ and without restricting $P$ to any specific form; therefore, the constraint gets satisfied up to a numerical tolerance (see Section \ref{Sec: Unsteady Aerodynamics Model}).
\end{itemize}

\section{Numerical Examples} \label{sec:Numerical Examples}
Numerical results for three examples are provided here.
Section \ref{Sec: Modified Lorenz–Stenflo System} discusses an example with a generalized lossless nonlinearity.
The examples in Sections \ref{Sec: Unsteady Aerodynamics Model} and \ref{Sec: Two-dimensional Academic Example}
have nonlinearities that are lossless in the standard sense,
and the proposed procedure of computing the shift coordinates leads to tighter trapping regions for these two examples compared to prior methods. 
All the results in this section are computed in MATLAB R2025b using cvx with SDPT3 as the solver.
We take $\epsilon = 10^{-6}$ and set \texttt{cvx\_precision} to `best' to promote high numerical accuracy.
The polynomial computations (see Section \ref{Sec: Analysis Conditions}) are carried out using the `collect' and `poly2basis' routines in SOSTOOLS \cite{sostools}.
These computations lead to $G$, which is fixed for a given quadratic nonlinearity. For the Lorenz attractor in \eqref{eq:Lorenz_system}, $G$ is a $9 \times 5$ matrix of zeroes with the non-zero entries
\begin{align*}
    G(4,3) &= G(5,4) = G(6,5) = -1, \\ 
    G(7,1) &= G(8,2) = G(9,4) = 1.
\end{align*}
Then, enforcing $G^\top \mathrm{vec}(S) = 0$ yields the  $S$ in \eqref{eq:generalized_lossless_matrix_Lorenz}.
The largest principal semi-major axis length is used to determine ellipsoid `size'. For a trapping ellipsoid $\mathcal{E}(P,r)$, it is given by $\alpha := \frac{r}{\sqrt{\lambda_{\mathrm{min}}(P)}}$ \cite{pope2008ellipsoids}.
We refer to $\alpha$ as the ``spherical trapping radius'', which facilitates comparison with existing spherical trapping region results.
Note that $\mathcal{E}(P,r) \subseteq \mathcal{B}(\alpha)$ as $\mathcal{B}(\alpha)$ circumscribes $\mathcal{E}(P,r)$. 
Thus, trajectories will eventually converge to $\mathcal{B}(\alpha)$, which is not necessarily a trapping region as it might not be invariant. 
Trapping ellipsoids are computed on a $\chi$ grid of 100 logarithmically spaced points between 0.1 and 10. 
The GEVP \eqref{eq:trapping_optimization_GEVP} is solved using a bisection on $r$.
The codes are available at \url{https://github.com/digantabhattacharjee/Ellipsoidal-global-trapping}


\subsection{Modified Lorenz–Stenflo (MLS) System} \label{Sec: Modified Lorenz–Stenflo System}
We modify the four-dimensional Lorenz-Stenflo system given in \cite{peng2025extending} to the following form 
\begin{equation} \label{eq:Modified_Lorenz–Stenflo_system}
\begin{aligned}
\dot{x} &= \sigma (y-x) + \gamma w, \
\dot{y} = \mu x - xz - y \\ 
\dot{z} &= 2xy - \kappa z, \
\dot{w} = -x - \sigma w
\end{aligned} 
\end{equation}
where $\sigma = 2$, $\kappa = 0.7$, $\mu=26$ and $\gamma=1.5$.
The nonlinearity in \eqref{eq:Modified_Lorenz–Stenflo_system} is \textit{not} lossless as needed for existing methods, but admits a \emph{generalized lossless} constraint $G^\top \mathrm{vec}(P) = 0$ with respect to any $P=P^\top$ given by 
\begin{equation*}
    P = \begin{bmatrix}
        p_{11} & 0          &  0  & p_{14}  \\
         0     &  2 p_{33}  &  0  & 0       \\
         0     &     0      &  p_{33}  & 0 \\
       p_{14}  &     0  &      0    & p_{44}
    \end{bmatrix},~p_{ij} \in \mathbb{R}.
\end{equation*} 
Using the above structure of $P$, we solve 
\eqref{eq:shift_optimization_initial} with $\delta_m = 30$ to yield $m^\star = \begin{bmatrix}
0 & 0 & 25.22 & 0 \end{bmatrix}^\top$,  $a^\star = -4.91$, and $(p^\star_{11},p^\star_{14},p^\star_{33},p^\star_{44}) = (6.67,3.98,5,5.26)$.
Hence, \eqref{eq:Modified_Lorenz–Stenflo_system} is globally bounded and admits a trapping ellipsoid.

The spherical trapping radii on the $\chi$ grid are shown in Fig. \ref{fig:Modified_Lorenz–Stenflo system}.
The best solution on the grid corresponds to $(r^\star, \chi^\star) = (25.743, 0.559)$, $(p^\star_{11},p^\star_{14},p^\star_{33},p^\star_{44}) = (2.16,2.63,1,7.01)$, and $\alpha^\star = r^\star$. 
We observe marginal improvements by solving the GEVP \eqref{eq:trapping_optimization_GEVP}, which yields $(r^{\star^{[1]}}, \chi^{\star^{[1]}}) = (25.721, 0.562)$ and $\alpha^{\star^{[1]}} = r^{\star^{[1]}}$.
The local search in the shift coordinates (see Section \ref{sec:Numerical Solutions}) does not provide any improvement for this example.  
Finally, the estimated ultimate bound from the GEVP solution is $\alpha^{\star^{[1]}} + \|m^\star\| = 50.94$.
Running 1000 simulations of the system \eqref{eq:Modified_Lorenz–Stenflo_system}, with each state randomly initialized between -100 and 100 for each run, yields an empirical ultimate bound of 42.34. 
Thus, the estimate obtained from the analysis is somewhat conservative.

\begin{figure}[!hbt]
\centering
\includegraphics[width=1\columnwidth]{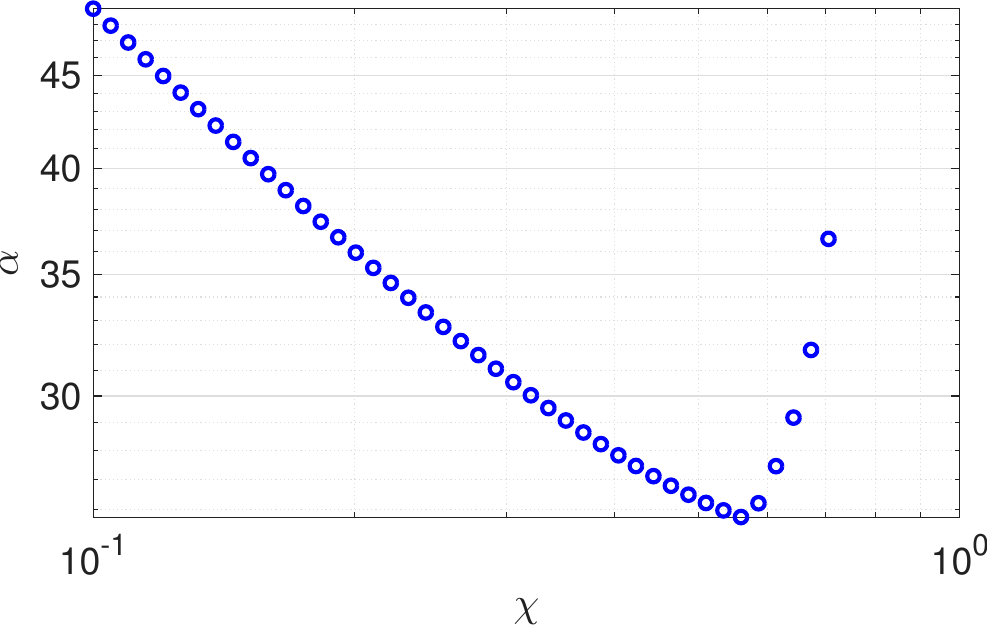}
\caption{Radii $\alpha$ over $\chi$ grid for the MLS system \eqref{eq:Modified_Lorenz–Stenflo_system}. The trend is qualitatively similar for the unsteady aerodynamics model~(Section \ref{Sec: Unsteady Aerodynamics Model}) and the academic example (Section \ref{Sec: Two-dimensional Academic Example}). } 
\label{fig:Modified_Lorenz–Stenflo system}
\end{figure}

\subsection{Unsteady Aerodynamics Model~\cite[Appendix C]{heide2025data}} \label{Sec: Unsteady Aerodynamics Model}
We apply the proposed method on a low-order model of unsteady separated flow over a NACA 65(1)-412 airfoil at a Reynolds number of 20,000 \cite{heide2025data}.
The system is globally bounded \cite{heide2025data} and solving the optimization in \eqref{eq:shift_optimization_initial} with $\delta_m = 50$ yields $a^\star = -5.81 \times 10^{-4}$ and $m^\star = \begin{bmatrix}
-18.6 & 8.3 & -2.3 & -14.6 & -42.7 & -6.5
\end{bmatrix}^\top$.
The best solution on the $\chi$ grid corresponds to $(r^\star,\chi^\star) = (4.308, 4.535)$, $\alpha^\star = 430.8$ and $P^{\star^{[1]}} \approx 10^{-4}I$ (details omitted here for brevity).
We get minimal improvements by solving the GEVP \eqref{eq:trapping_optimization_GEVP}, which yields $(r^{\star^{[1]}},\chi^{\star^{[1]}}) = (4.306, 4.457)$ and $\alpha^{\star^{[1]}} = 430.6$.
We take the soft constraint approach for $G^\top \mathrm{vec}(P) = 0$ in this example, and the largest (absolute value) element of $G^\top \mathrm{vec}(P^{\star^{[1]}})$ equals $7.85 \times 10^{-10}$.
The local search in the shift coordinates does not lead to any further improvements.
By contrast, the tightest trapping region obtained by employing the framework in \cite{liao2025convex} on the same model is $\mathcal{B}(1637.56)$, and the corresponding shift coordinates are $m = \begin{bmatrix}
        -205 &  75  &  -27 & -112 & -266  & -151 
    \end{bmatrix}^\top$.
The smaller trapping region in the current analysis is largely due to the constraint $\|m\| \leq \delta_m$ in \eqref{eq:shift_optimization_initial}, and the value of 
$\delta_m$ can be tuned for further improvements.   
Finally, the ultimate bound estimated via the framework in \cite{liao2025convex} is 2031.3, while the proposed method yields a much tighter estimate of $\alpha^{\star^{[1]}} + \|m^\star \| = 480.6$.

\subsection{Two-dimensional Academic Example \cite{liao2025convex}} \label{Sec: Two-dimensional Academic Example}
The two-dimensional system~\cite{liao2025convex}
\begin{equation*} 
\begin{bmatrix}
\dot{x}_1 \\ \dot{x}_2 
\end{bmatrix} = \begin{bmatrix}
-1  & 0 \\ 
0 & -4
\end{bmatrix} \begin{bmatrix}
x_1 \\ x_2 
\end{bmatrix} + \begin{bmatrix}
-x_1 x_2 \\ x_1^2
\end{bmatrix} + \begin{bmatrix}
0 \\ 1
\end{bmatrix}
\end{equation*}
has a globally asymptotically stable equilibrium point at $x = \begin{bmatrix}
    0 & 0.25
\end{bmatrix}^\top$.
Here, we explicitly enforce $G^\top \mathrm{vec}(P) = 0$, which yields $P=I$, reverting to the standard lossless setting.
We skip the initial shift coordinate computation using \eqref{eq:shift_optimization_initial} and simply take $m^\star = 0$.
The best solution on the grid is given by $(m^{\star^{[1]}}, P^{\star^{[1]}}, r^\star,\chi^\star) = (0, I, 0.2905, 1.963)$.
Upon solving the GEVP \eqref{eq:trapping_optimization_GEVP}, we obtain $(r^{\star^{[1]}},\chi^{\star^{[1]}}) = (0.289, 2)$.
This means that the trapping region is $\mathcal{B}(0.289)$, which is the same as in \cite{liao2025convex}. 
However, the local search in the shift coordinates yields $m^{\star^{[2]}} = \begin{bmatrix}
    0 & 0.25
\end{bmatrix}^\top$ and $(P^{\star^{[2]}}, r^{\star^{[2]}}) = (I,  7.07 \times 10^{-4})$. 
Therefore, the local search not only computes a tighter trapping region, but also (approximately) identifies the equilibrium point of the system.
We take $(m^{\star^{[2]}}, P^{\star^{[2]}}, r^{\star^{[2]}}, \chi^{\star^{[1]}}) = (\begin{bmatrix}
    0 & 0.25
\end{bmatrix}^\top, I, 7.07 \times 10^{-4}, 2)$ as the final solution.

\begin{table}[!hbt]
\centering
\caption{Comparison of spherical trapping radii}
\begin{tabular}{lllll}
\cline{1-5}
\multicolumn{1}{|l|}{Example} & \multicolumn{1}{l|}{ Method 1}  & \multicolumn{1}{l|}{ \thead{ Method 2} } & \multicolumn{1}{l|}{ \thead{ Method 3} } & \multicolumn{1}{l|}{\thead{Proposed}} \\ \cline{1-5}
\multicolumn{1}{|l|}{MLS} & \multicolumn{1}{l|}{N/A} & \multicolumn{1}{l|}{N/A} & \multicolumn{1}{l|}{N/A} & \multicolumn{1}{l|}{25.721} \\ \cline{1-5}
\multicolumn{1}{|l|}{Airfoil} & 
\multicolumn{1}{l|}{1637.56} & 
\multicolumn{1}{l|}{1649.63} & \multicolumn{1}{l|}{1649.63} & \multicolumn{1}{l|}{430.6} \\ \cline{1-5}
\multicolumn{1}{|l|}{Academic} & \multicolumn{1}{l|}{0.289} & 
\multicolumn{1}{l|}{1} & \multicolumn{1}{l|}{1} & \multicolumn{1}{l|}{$0.0007$} \\ \cline{1-5}
\end{tabular}
\label{table:comparison_test_cases}
\end{table}

A comparative summary of the results provided in this section is outlined in Table \ref{table:comparison_test_cases}, where methods 1-3 denote prior work in \cite{liao2025convex}, \cite{schlegel2015long} and \cite{goyal2025guaranteed}, respectively.
Note that the shift coordinates computed using method 1 are used in methods 2 and 3.
For method 3, we utilize the expression given in Section \ref{Sec: Analysis Conditions} to avoid a typo in \cite{goyal2025guaranteed}.
Furthermore, Fig. \ref{fig:Phase_portrait_2D_example} shows a phase portrait of the two-dimensional academic example, along with the associated trapping regions provided in Table \ref{table:comparison_test_cases}. 
It illustrates the improvement offered by the proposed approach for this example relative to the existing methods.

 \begin{figure}[!hbt]
 \centering
 \includegraphics[width=1\columnwidth]{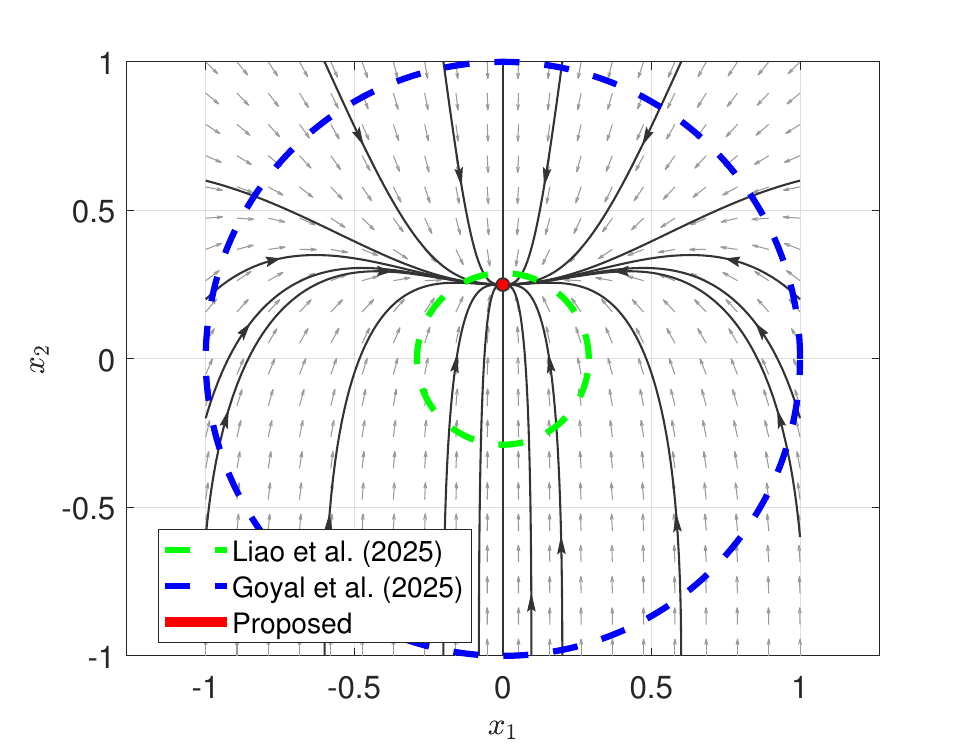}
 \caption{Phase portrait for the academic example (Section \ref{Sec: Two-dimensional Academic Example}), along with the corresponding trapping regions given in Table \ref{table:comparison_test_cases}. The method in \cite{schlegel2015long} yields the same trapping region as obtained via the method by Goyal et al. \cite{goyal2025guaranteed}.} 
 \label{fig:Phase_portrait_2D_example}
 \end{figure}


\section{Conclusions} \label{sec:Conclusions}
We propose a framework for trapping region analysis rooted in an efficient parameterization and optimization of the generalized lossless property in quadratic systems.
Our numerical procedure offers ways of optimizing the shift coordinates to compute tighter trapping regions.
Three examples are provided, highlighting key features like reduced conservatism in the ultimate bound (when applicable) and broader applicability than the prior methods.
Future work will study local trapping regions with the lossless constraint relaxed and draw parallels with prior boundedness analyses.

\bibliographystyle{IEEEtran}
\bibliography{Bibliography.bib}

@article{schlegel2015long,
  title={On long-term boundedness of {Galerkin} models},
  author={Schlegel, Michael and Noack, Bernd R},
  journal={J. Fluid Mech.},
  volume={765},
  pages={325--352},
  year={2015},
  publisher={Cambridge University Press}
}

@article{kaptanoglu2021promoting,
  title={Promoting global stability in data-driven models of quadratic nonlinear dynamics},
  author={Kaptanoglu, Alan A and Callaham, Jared L and Aravkin, Aleksandr and Hansen, Christopher J and Brunton, Steven L},
  journal={Phys. Rev. Fluids},
  volume={6},
  number={9},
  pages={094401},
  year={2021},
  publisher={APS}
}

@article{liao2025convex,
  title={A convex optimization approach to compute trapping regions for lossless quadratic systems},
  author={Liao, Shih-Chi and Leonid Heide, A and Hemati, Maziar S and Seiler, Peter J},
  journal={Int. J. Robust Nonlinear Control},
  volume={35},
  number={6},
  pages={2425--2436},
  year={2025},
  publisher={Wiley Online Library}
}

@book{khalil2002nonlinear,
  title={Nonlinear systems},
  author={Khalil, Hassan K},
  year={2002},
  publisher={Prentice Hall, Upper Saddle River, NJ}
}

@article{peng2025extending,
  title={Extending the trapping theorem to provide local stability guarantees for quadratically nonlinear models},
  author={Peng, Mai and Kaptanoglu, Alan A and Hansen, Christopher J and Stevens-Haas, Jacob and Manohar, Krithika and Brunton, Steven L},
  journal={Phys. Fluids},
  volume={37},
  number={10},
  year={2025},
  publisher={AIP Publishing}
}

@inproceedings{ames2019control,
  title={Control barrier functions: Theory and applications},
  author={Ames, Aaron D and Coogan, Samuel and Egerstedt, Magnus and Notomista, Gennaro and Sreenath, Koushil and Tabuada, Paulo},
  booktitle={Eur. Contr. Conf.},
  pages={3420--3431},
  year={2019},
}

@article{heide2025data,
  title={Data-driven nonlinear aerodynamics models with certifiably optimal boundedness properties},
  author={Heide, A Leonid and Liao, Shih-Chi and Castiblanco-Ballesteros, Sergio and Jacobs, Gustaaf B and Seiler, Peter and Hemati, Maziar S},
  journal={arXiv preprint arXiv:2508.16800},
  year={2025}
}

@article{goyal2025guaranteed,
  title={Guaranteed stable quadratic models and their applications in {SINDy} and operator inference},
  author={Goyal, Pawan and Duff, Igor Pontes and Benner, Peter},
  journal={Phys. D: Nonlinear Phenom.},
  pages={134893},
  year={2025},
  publisher={Elsevier}
}

@techreport{pope2008ellipsoids,
  author      = "Stephen B. Pope",
  title       = "Algorithms for Ellipsoids",
  institution = "Cornell University, FDA-08-01",
  year        = "2008"
}

@article{boyd1993method,
  title={Method of centers for minimizing generalized eigenvalues},
  author={Boyd, Stephen and El Ghaoui, Laurent},
  journal={Linear Algebra Appl.},
  volume={188},
  pages={63--111},
  year={1993},
  publisher={Elsevier}
}

@manual{sostools,
author = {A. Papachristodoulou and J. Anderson and G. Valmorbida and S. Prajna and P. Seiler and P. A. Parrilo},
title = {{SOSTOOLS}: Sum of squares optimization toolbox for {MATLAB}},
year = {2013}
}

@article{tan2008stability,
  title={Stability region analysis using polynomial and composite polynomial Lyapunov functions and sum-of-squares programming},
  author={Tan, Weehong and Packard, Andrew},
  journal={IEEE Trans. Automat. Contr.},
  volume={53},
  number={2},
  pages={565--571},
  year={2008},
  publisher={IEEE}
}

@article{chesi2005lmi,
  title={{LMI}-based computation of optimal quadratic Lyapunov functions for odd polynomial systems},
  author={Chesi, Graziano and Garulli, Andrea and Tesi, Alberto and Vicino, Antonio},
  journal={Int. J. Robust Nonlinear Control},
  volume={15},
  number={1},
  pages={35--49},
  year={2005},
  publisher={Wiley Online Library}
}

@article{liu2020input,
  title={Input-output inspired method for permissible perturbation amplitude of transitional wall-bounded shear flows},
  author={Liu, Chang and Gayme, Dennice F},
  journal={Phys. Rev. E.},
  volume={102},
  number={6},
  pages={063108},
  year={2020},
  publisher={APS}
}

@article{cavalieri2022transition,
  title={Transition to chaos in a reduced-order model of a shear layer},
  author={Cavalieri, Andr{\'e} VG and Rempel, Erico L and Nogueira, Petr{\^o}nio AS},
  journal={J. Fluid Mech.},
  volume={932},
  pages={A43},
  year={2022},
  publisher={Cambridge University Press}
}

@book{vidyasagar2002nonlinear,
  title={Nonlinear systems analysis},
  author={Vidyasagar, Mathukumalli},
  year={2002},
  publisher={SIAM}
}

@article{loiseau2018constrained,
  title={Constrained sparse {Galerkin} regression},
  author={Loiseau, Jean-Christophe and Brunton, Steven L},
  journal={J. Fluid Mech.},
  volume={838},
  pages={42--67},
  year={2018},
  publisher={Cambridge University Press}
}

@article{callaham2023multiscale,
  title={Multiscale model reduction for incompressible flows},
  author={Callaham, Jared L and Loiseau, Jean-Christophe and Brunton, Steven L},
  journal={J. Fluid Mech.},
  volume={973},
  pages={A3},
  year={2023},
  publisher={Cambridge University Press}
}

@article{balajewicz2013low,
  title={Low-dimensional modelling of high-{Reynolds}-number shear flows incorporating constraints from the {Navier--Stokes} equation},
  author={Balajewicz, Maciej J and Dowell, Earl H and Noack, Bernd R},
  journal={J. Fluid Mech.},
  volume={729},
  pages={285--308},
  year={2013},
  publisher={Cambridge University Press}
}

@article{kraichnan1989there,
  title={Is there a statistical mechanics of turbulence?},
  author={Kraichnan, Robert H and Chen, Shiyi},
  journal={Phys. D: Nonlinear Phenom.},
  volume={37},
  number={1-3},
  pages={160--172},
  year={1989},
  publisher={Elsevier}
}

@article{kaptanoglu2021physics,
  title={Physics-constrained, low-dimensional models for magnetohydrodynamics: First-principles and data-driven approaches},
  author={Kaptanoglu, Alan A and Morgan, Kyle D and Hansen, Chris J and Brunton, Steven L},
  journal={Phys. Rev. E.},
  volume={104},
  number={1},
  pages={015206},
  year={2021},
  publisher={APS}
}

@book{schmid_henningson_book,
  author = {Schmid, P. J. and Henningson, D. S.},
  doi = {10.1007/978-1-4613-0185-1},
  publisher = {Springer},
  title = {{Stability and Transition in Shear Flows}},
  year = 2001
}


\end{document}